\documentclass[11pt]{article}
\usepackage{amsmath, amsthm, amssymb}
\usepackage{graphicx}
\usepackage[latin1]{inputenc}

\usepackage{enumerate}
\usepackage{setspace}
\usepackage{url}
\usepackage{epstopdf}
\usepackage{multirow}
\usepackage{amsmath}
\usepackage[autolanguage]{numprint}
\usepackage{algorithm}
\usepackage[noend]{algorithmic}
\usepackage{amsmath, amsthm, amssymb}

\usepackage{mathtools}
\DeclarePairedDelimiter{\ceil}{\lceil}{\rceil}

\DeclarePairedDelimiter{\parens}{\lparen}{\rparen}

\def\maxsumsextended{\textsc{Maximum Consecutive Subsums Problem}}
\def\maxsums{\textsc{MCSP}}
\def\E{\operatorname{E}}

\def\Prob{\operatorname{Pr}}

\begin{document}

%\title{Elsevier \LaTeX\ template\tnoteref{mytitlenote}}
%%% Group authors per affiliation:

%\begin{keyword}
%\texttt{elsarticle.cls}\sep \LaTeX\sep Elsevier \sep template
%\MSC[2010] 00-01\sep  99-00
%\end{keyword}

%\linenumbers

\newtheorem{conjecture}{Conjecture}
\newtheorem{theorem}{Theorem}
\newtheorem{proposition}{Proposition}

\title{Searching for a superlinear  lower bounds for the Maximum Consecutive Subsums Problem and the  (min,+)-convolution}
%\date{}

\author{Wilfredo Bardales Roncalla\\  PUC-RIO, Brazil\\
{\tt laber@inf.puc-rio.br} \and Eduardo Laber\\ PUC-RIO, Brazil\\
{\tt laber@inf.puc-rio.br} \and  Ferdinando Cicalese \\Department of Computer Science, University of Verona, Italy \\
{\tt ferdinando.cicalese@univr.it}}

%\author[1]{Wilfredo Bardales Roncalla}
%\author[1]{Eduardo Laber}
%\author[2]{Ferdinando Cicalese}

%\address[1]{Departamento de Inform\'atica, Pontificie Universidade Cat\'olica do Rio de Janeiro, Brazil}
%Rio de Janeiro, Rio de Janeiro 22451-900, Brazil\\
%{\tt wilfre.br@gmail.com}  \and Eduardo Laber \\PUC-Rio, Brazil\\ {\tt laber@inf.puc-rio.br}
%\address[2]{Department of Computer Science, University of Verona, Italy}
%{\tt cicalese@dia.unisa.it}}

\maketitle

\begin{abstract}
Given a sequence of $n$ numbers, the \maxsumsextended (\maxsums) asks for the maximum consecutive sum of lengths $\ell$ for each $\ell = 1, \dots, n$. No algorithm is known for this problem which is significantly better than the naive quadratic solution. 
Nor a super linear lower bound is known. 
The best known bound for the \maxsums\ is based on the the computation of the $(min,+)$-convolution, another problem for which 
neither an $O(n^{2-\epsilon})$ upper bound is known nor a super linear lower bound. 
% In fact, the two problems also share the fact that no 
%lower bound is known which is better than the trivial $O(n).$ 
We show that the two problems are in fact computationally equivalent by providing linear
reductions between them. Then, we concentrate on the problem of finding super linear lower bounds and provide empirical evidence for an $\Omega(n \log n)$ lower bounds for both problems in the decision tree model. 
\end{abstract}

\section{Introduction}

%There exists a long line of research on identifying or selecting maximal consecutive subsums in a numerical sequence $A=(a_1, \dots, a_n)$. 
%The classical example is perhaps the \textsc{Maximum Sum Segment} problem (\textsc{MSS})
%asking for the consecutive subsequence attaining the maximum among all possible ${n \choose 2}+n$ subsequences. The problem 
%was introduced by Grenader~\cite{Bentley:1986:PP} and has found applications in  
%pattern matching~\cite{Grenander:1978:PA}, biological sequence analysis~\cite{Allison:2003:LBILNSI}, and data mining~\cite{Fukuda:2001:DMO}. The \textsc{MSS} can be solved in linear time using Kadane's algorithm~\cite{Bentley:1986:PP}. 
%
%The \textsc{$k$ Maximum Sum Segments} ($k$-\textsc{MSS}), asking to find the $k$ largest consecutive subsums in a sequence, has also been investigated 
%\cite{BaeTakaoka:2006:IAK,ChengChengTienChao:2005:IAK,LinLee:2007:RAS}  and eventually proved to be solvable in $O(n+k)$ in~\cite{Brodal:2007:LTA}. A linear time algorithm 
%exists also for the case where the disjoint segments of maximum sum have to be found \cite{RuzzoTompa:1999:LTAFAMSS}, 
%and for the variants of \textsc{MSS} and $k$-\textsc{MSS} where only segments whose length is within a given interval have to be considered. 
%
%All these problems admit a linear solution. In contrast, notwithstanding the apparent similarity,
%the following variant known as the Maximum Consecutive Subsums Problem (\maxsums\ ) has defied so far any attempt to provide a strongly sub-quadratic solution. 

Given a sequence $A = ( a_1, a_2, \ldots, a_n)$ of $n$  numbers, the \maxsums\ asks to compute
%{\bf for all} $\ell = 1, \dots, n,$ the maximum consecutive subsum of length $\ell,$ i.e.,
 the sequence
$m_1, \dots, m_n,$ where  
%\begin{equation} \label{eq:definition}
$m_{\ell} = \max_{i = 1, \dots n-\ell+1} \{ a_{i} + \dots + a_{i+\ell-1}\},$
%\end{equation}
is the  maximum over all consecutive subsums of length $\ell.$
%$A[i^*, j^*] = (a_{i^*}, a_{i+1}, \ldots, a_{j^*})$  of size $\ell = j^* - i^* + 1$ whose sum $\sum_{k = i^*}^{j^*}a_k$ is the maximum, for all $\ell = 1, 2, \ldots, n$.

The \maxsums\ appears in several scenarios of both theoretical and practical interest like 
approximate pattern matching~\cite{Burcsi:2012:OAJPMS}, mass spectrometry data analysis~\cite{Cieliebak},  and in the problem of 
locating large empty regions in data sets~\cite{Bergkvist:2005:FAFDSED}. Most work has been done for the case where the input sequence is binary, since in this case the 
 \maxsums\ coincides with the problem of constructing
membership query indexes  for  jumbled pattern matching 
\cite{Burcsi:2012:OAJPMS,Golnaz:2012:BJSM,Travis:ESA13}.

%%%%

%even in the case of binary sequences.
It is not difficult to come up with simple $O(n^2)$ solutions for the \maxsums\, since each value $m_{\ell}$ %in (\ref{eq:definition})
can be easily computed in 
linear time by  one pass over the input sequence. Surprisingly, despite the growing interest generated by this problem 
(see, e.g., ~\cite{Burcsi:2012:OAJPMS,Cicalese:2009:CFL,Cicalese:2012:AMCSP,Moosa:2012:STL}, and references 
therein quoted), no solution is known with running time  $O(n^{2-\epsilon})$ for some constant $\epsilon>0$, 
nor  a lower bound better than the trivial $\Omega(n)$ is known.

%A better complexity can be achieved for the  \maxsums\ on 0/1 sequences with long runs of the same character, 
%i.e., for sequences admitting a good  run-length
%compression  \cite{Golnaz:2012:BJSM,Travis:ESA13}.

Algorithms that produce approximate solutions for the \maxsums\ are also known \cite{Cicalese:2012:AMCSP}.
However, for the general case, the best  available algorithm in the real RAM model runs
in $O(n^2/\log n)$  \cite{Burcsi:2012:OAJPMS,Moosa:2012:STL}
and makes use of the algorithm  proposed in \cite{conf/esa/BremnerCDEHILT06} for computing a $(\min, +)$-convolution.

The $(min,+)$-convolution problem is a natural variation of the classical convolution problem:
% \cite{books/daglib/0015106} with several applications in sequence alignment \cite{Epp89}
% and sequential data analysis \cite{conf/nips/FelzenszwalbHK03}. 
 Given two sequences $X=(x_0,x_1,\ldots,x_n)$ and $Y=(y_0,y_1,\ldots,y_n)$ of real numbers,
the  $(min,+)$-convolution of $X$ and $Y$ is the sequence $z = z_0, z_1 \dots z_{2n},$ with 
$ z_k= \min_{\atop i=0,\ldots,k} \{ x_i+y_{k-i} \}$, for $ k=0,\ldots,2n.$

This problem   has important applications in a variety of areas, including signal processing, 
pattern recognition, computer vision, and mathematical programming. 
According to  \cite{conf/esa/BremnerCDEHILT06}, this problem has appeared frequently 
 in the literature since Bellman's early work on dynamic programming. 

Like for the  \maxsums, no strongly subquadratic algorithm appears to be known to 
compute the $(min, +)$-convolution.
The  best known algorithm for computing $(min, +)$-convolution runs in $O(n^2/\log n)$.
Recently, Williams \cite{Williams} has provided a Monte Carlo algorithm that  computes the  
 $(min, +)$ convolution  in $O(n^2 / 2^{\Omega(\log^{1/2} n)})$ time on the real RAM. Although, this bound is 
better than $O(n^2/polylog(n))$ it is still $\omega(n^{2-\epsilon})$ for any $\epsilon > 0.$
For the special case of monotone increasing sequences with elements bounded by $O(n)$ a recent breakthrough 
in \cite{ChanSTOC2015} shows that $(min, +)$-convolution can be computed in time $\tilde{O}(n^{1.864}).$ 
This result implies an equivalent sub quadratic bound for \maxsums\ on 0/1 sequences. 

Taking into account the apparent difficulty to devise an $O(n^{2-\epsilon})$ 
algorithm for the \maxsums\ and for computing the  $(min,+)$-convolution,
 a natural question  to ask is whether there exists a non-trivial superlinar lower bound 
 for these problems. 
%\footnote{An analogous perspective---i.e., trying to improve the only known trivial lower bound---has been taken very recently in 
%\cite{ChanICALP2014} where the authors have studied the 3-SUM hardness of indexing for  jumbled pattern matching, a problem that in the case of binary strings coincides to the MCSP over 0-1 sequences.}.

The decision tree model is a widely used model of computation to study lower bounds for algorithmic problems.
Fundamental algorithmic problems as searching, sorting and selection are examples of problems that have been studied in this model.
In the decision tree model each algorithm can be
represented by a decision tree, where internal nodes correspond to computations; the branches(edges) that leave an internal node $v$  correspond to the  possible results of the computation associated with $v$ and the leaves correspond to the possible outputs of the algorithm (see section \ref{sec:LowerBound} for the exact definitions). Some variants/specializations 
of the decision tree model as linear decision trees and algebraic decision trees have also been used to study the complexity
of algorithmic problems \cite{BenOr83a}. In \cite{conf/esa/BremnerCDEHILT06} it is shown that the $(min, +)$-convolution of two $n$-element vectors 
can be computed in $O(n^{3/2})$ in the non-uniform linear decision tree model. 

This work describes our findings in the quest for a super-linear lower bound 
for the \maxsums\ in the decision tree model  of computation. 

\medskip

\noindent
{\bf Our Contributions.}
%In Section \ref{}  we argue that a superlinear  lower bound for the \maxsums\  on the decision tree model of computation
%can be obtained  by constructing a large enough set of inputs ${\cal A}$ for
%\maxsums\ such that the intersection between the set of ouputs of instances $A$ and $A'$ is empty, for
%every pair $A,A' \in {\cal A}$.
%Athough we managed to 
%construct sets of exponential size we fell short
%to obtain a set of superexponential size
We provide computational  evidences that the running time of both  \maxsums\ and  (min,+)-convolution problem is
$\Omega(n \log n)$ in the decision tree model of computation. 

We first establish a linear equivalence between both problems  by showing a  linear reduction from 
 (min,+)-convolution to \maxsums\. The opposite direction 
was  shown in \cite{Bergkvist:2005:FAFDSED}.
As a result of this equivalence, any bound for one problem also holds for the other. Then, in the following, we only concentrate on the 
\maxsums\ . 
%
%Thus any bound shown for one problem also holds for the other. 
%In the following, we focus  our discussion on the  \maxsums.

In Section \ref{sec:LowerBound},  we argue that a lower bound for the \maxsums\  in the decision tree model 
can be obtained by generating a large set of inputs sequences  such that no pair of them produces
a common output. We present  a construction which shows that there exists such a set of exponential size, although, this is still 
not enough to  prove a superlinear lower bound on \maxsums\..

%This scenario stimulated us to   employ a computational approach to give
%evidences of the existence of a super-linear lower bound for \maxsums\..
%This is the topic of  Section \ref{sec:Experimental}.
In Section \ref{sec:Experimental}, by using a deterministic approach we show empirically that for 
$n \leq 14$ there exists a set of inputs, with the above property, 
and cardinality  larger than $(n/2)!$ so that  $n/2 \log (n/2)$ is a lower bound on the depth of any decision tree that solves the \maxsums\  for instances of size $n \leq 14$. This required 
27 hours of CPU time in our computational environment. In order to address larger values of $n$, we employed  sampling strategies. We devised a hypothesis test and showed  that the $n/2 \log (n/2)$ lower bound also holds for any $n \leq 100$ with confidence 
much larger than $99.999\%$.

%Although our initial goal has not been fulfilled yet, 
Our results can bring new insight on the complexity of both the \maxsums\ and the (min,+)-convolution problem and
may
represent a  initial step towards proving a superlinear lower bound for these problems. Moreover, the techniques employed 
might be useful in the  investigation of  lower bounds for other computational problems with the
same flavour.

\section{An Equivalence between \maxsums\ and the $(min,+)$-Convolution Problem}
\label{sec:convolution}

We start by showing that  \maxsums\ and $(min,+)$-convolution are computationally equivalent.

\begin{theorem}
There exist linear reductions between  the \maxsums\ and  the  $(min,+)$-convolution. 
\end{theorem}
\begin{proof}
The reduction from \maxsums\, to $(min, +)$-convolution is observed in \cite{Bergkvist:2005:FAFDSED} where the latter problem is
presented under the geometric naming Lowest Midpoints Problem.
Thus, we just present the reduction from the $(min, +)$-convolution to \maxsums.

Let
$I=(X,Y)$ be an input for the $(min,+)$-convolution, where $X=(x_0,\ldots,x_n)$ and $Y=(y_0,\ldots,y_n)$.
Let $S$ be a  large enough number and define the input sequence  $A=(a_1,\ldots,a_{2n+4})$ 
for the \maxsums\ as follows: $a_{n+1} = a_{n+4} = S$; $a_{n+3} = -y_0$; $a_{n+2} = -x_0$; 
for each $i < n+1$ set $a_i = x_{n-i} - x_{n+1-i}$ and for $i > n+4$ set $a_i = y_{i-n-5} -y_{i-n-4}:$

\begin{table}[h]
\small
\centering
%\caption{The  number of unique configurations for $n=1,\ldots,14$ compared to the value
%\label{tab:deterministicresults}
%$n/2!$.}
% between the number of unique configurations and $(n/2)!$}
\begin{tabular}{cccccccccccc}
%\hline
$a_1$ & $\cdots$ & $a_n$ & $a_{n+1}$ & $a_{n+2}$ & $a_{n+3}$ & $a_{n+4}$ & $a_{n+5}$ & $a_{n+6}$ & $\cdots$ & $a_{2n+4}$\\
$(x_{n-1} - x_n)$  & $\cdots$ & $(x_0 - x_1)$ &
$S$ & $-x_0$ & $-y_0$ & $S$ &
$(y_{0} - y_1)$ & $(y_{1} - y_{2})$ & $\cdots$ & $(y_{n-1} - y_n)$\\
%\hline
\end{tabular}
\end{table}

Assuming $S > \sum_{i=0}^n |x_i| + |y_i||,$ for any $k \geq 4$ 
a maximum sum of $k$ consecutive elements includes $a_{n+1}$ and $a_{n+4}.$ Then,  for each $k \geq 4,$ 
$$\max_{j=1, \dots, (2n+4) -k+1} \sum_{i=j}^{j+k-1} a_i = 
\max_{\substack{0 \leq s,t \leq k-4 \\  s+t = k-4}} \sum_{i=n+1-t}^{n+4+s} a_i =
\max_{\substack{0 \leq s',t' \leq k-4 \\ s'+t' = k-4}}  2S -x_{t'} - y_{s'} =$$
$$2S - \min_{\substack{0 \leq t' \leq k-4}}  (x_{t'} + y_{k-4-t'})  = 2S - z_{k-4}.$$ 
where, as above, $z_0, \dots, z_{2n}$  denotes the result of the convolution of $X$ and $Y$. 
Therefore, for each $k=0, \dots, 2n,$ we have the equivalence
$z_{k}= 2S - \left ( \sum_{j=p_{k+4}}^{p_{k+4}+k+3} a_i \right ),$
where $p_k$ is the starting position of a maximum consecutive sum of length $k$ in  $A$. 
 \end{proof}
 
Due to these linear time reductions we can conclude that \maxsums\, and $(min, +)$-convolution have the same time complexity.
In the following, we will focus our  discussion on lower bounds for \maxsums\, and any conclusion reached will also 
hold for the $(min, +)$-convolution.

%Hence,  if ${\cal A}=\{A_1,\ldots,A_p\}$ is a pairwise
%disjoint set of inputs for \maxsums\ then ${\cal I}=\{I_1(A),\ldots,I_p(A)\}$,
% is a set of pairwise
%disjoint inputs for  the  $(min,+)$-convolution problem.
%As a consequence, the results presented in the previous sections also holds for the
%$(min,+)$-convolution problem, that is, with confidence larger than $99.99\%$ we
% need at least $n/2 \log (n/2)$ steps to solve the $(min,+)$-convolution problem
% for any $n \leq 100$.

%constructing a large enough set of inputs ${\cal A}$ such that
%can be proved by showing a lower bound on a minimum set of outputs that any algorithm that solves the \maxsums\ must know how to distinguish from. This set must be large enough such that a non-trivial lower bound may be derived using a decision tree model of computation. 

%This technique in principle can be applied to any kind of problem for which a set like that exists.

%Then, we show an algebraic way to explore this set. We describe some characteristics about its elements, and we prove that its size is at least exponential.

%Finally we give some empirical evidence, using deterministic and probabilistic methods, to show that the size of the aforementioned set is $\Omega(\frac{n}{2}!)$ which implies a superexponential size and a corresponding superlinear lower bound. 

\section{Towards a Lower Bound for the MCSP}
\label{sec:LowerBound}
In this section we discuss our approach  to prove a lower bound for the \maxsums.
It will be convenient to employ the following alternative
formulation of the \maxsums\:

\medskip

{\em Input.} A sequence $A=(a_1,\ldots,a_n)$ of $n$ real numbers;

\medskip

{\em Output.} A sequence $P=(p_1,\ldots,p_n)$,
where $p_{\ell}$, for $\ell =1,\ldots,n$, is the starting position of a consecutive subsequence of $A$  that
has maximum sum among the consecutive subsequences of $A$ with length $\ell$, i.e., 
the subsequence $a_{p_\ell}\,  a_{p_\ell + 1} \ldots a_{p_i + \ell - 1}$ is a maximum consecutive subsum of length $\ell$.

We call the sequence  $P$ an \emph{output configuration} or simply a
\emph{configuration}. 

For example, for the input sequence $A = (3, 0, 5, 0, 2, 4)$ the only output configuration is given by the sequence
$P = (5, 5, 1, 3, 2, 1),$ which says, e.g., that there is a maximum consecutive subsum of length $4$ starting at position  $3.$
%illustrated in Table~\ref{tab:mcspoutput}. 

%
%\begin{table}
%\centering
%\caption{Maximum consecutive subsums for the sequence $A=(3, 0, 5, 0, 2, 4)$}
%\label{tab:mcspoutput}
%\begin{tabular}{ c l r }
%\hline\noalign{\smallskip}
%$i$ & \multicolumn{1}{c}{$m_i$} & $p_i$ \\ 
%\noalign{\smallskip}
%\hline
%\noalign{\smallskip}
%$1$ & $A[3]$ & $5$ \\
%$2$ & $A[5, 6]$ & $5$ \\
%$3$ & $A[1\ldots 3]$ & $1$ \\
%$4$ & $A[3\ldots 6]$ & $3$ \\
%$5$ & $A[2\ldots 6]$ & $2$ \\
%$6$ & $A[1\ldots 6]$ & $1$ \\
%\hline
%\end{tabular}
%\end{table}

We note that there are $n!$ possible configurations because
 $p_i$, for $i=1,\ldots,n$, can assume any value in the set $\{1,\ldots,n-i+1\}$.
In particular for the input $A=(a_1,\ldots,a_n)$, with $a_1=a_2=\cdots=a_n=1$, all
the $n!$ possible configurations are output configurations for $A$.
This example also shows that some input sequences  have more than one output configuration.

We will study the above version of \maxsums\, in the {\em decision tree model}. An algorithm in this
model is a ternary tree. Each internal node $I$  contains a test $f(I) \lesseqgtr 0 ?$ for some 
rational function $f$ of $n$ (size of the input) arguments\footnote{If all functions
are restricted to be linear (resp.\ polynomials) the model is referred to as the 
{\em linear} (resp. algebraic) decision tree model.}. Each leaf of the tree 
contains a set of function $g_i$ ($i = 1, \dots, n$) on the $n$-valued input. For any input 
$A = (a_1, \dots, a_n)$, the algorithm moves from the root down the tree. At each 
node the corresponding test is performed and a branch is followed according to 
whether the test outcome is $>0,\, < 0, \, = 0.$ When a leaf is reached, the output
$P = (p_1, \dots, p_n)$ is given by $p_i = g_i(A).$ The cost of the algorithm is defined 
to be the height of the tree. The complexity $K(n)$ in this model is the minimum 
cost of any such algorithm. We will be concerned with an information theoretic lower bound on 
$K(n).$

\subsection{An approach based on unique configurations}

For any input $A$ for the \maxsums\ we define 
 $\mathcal{P}(A) = \{ P\ |\ P \mbox{ is an output configuration for } A \}$.
 
We say that a configuration $P$  is {\em unique}
if and only if there exists an input $A$ for the \maxsums\, for which ${\cal P}(A)=\{P\}$.

Our approach to prove a lower bound on \maxsums\ consists on
finding a large set of distinct unique configurations.

In fact, let $\mathbb{P} = \{P_1,P_2,\ldots,P_k\}$ be a set of distinct unique configurations.
Moreover, let $A_i$, for $i=1,\ldots,k$, be an instance for \maxsums\ such that $P_i$ is its
unique configuration, i.e., ${\cal P}(A_i) = \{P_i\}.$ 
Then, the instances in the set $\mathbb{A}=\{A_1,\ldots,A_k \}$ are distinct.
Hence, in any decision tree, for any distinct inputs $A_i, A_j$, the leaves associated to the corresponding
outputs must be distinct, hence the decision tree must have at least $k$ leaves.  
%In order to be able to solve the instances in ${\cal A}$ an algorithm must be able to 
%uniquely distinguish the configurations in ${\cal P}$. This requires 
Then, the height of the tree is at least $\ceil*{\log_3 k},$ which shows that 
% bits of information. Hence, 
$\ceil*{\log_3 k}$ is a lower bound to the problem in the decision tree model. 
%In fact, 
%the above considerations shows that we need at least that number of steps, or decisions, to uniquely identify 
%any solution from the set $\mathcal{P}$ using the decision tree model. 
We have proved the following.

\smallskip

\begin{theorem} \label{theo:lowerbound}
% \label{prop:lowerbound}
If $\mathbb{P}$ is a subset of distinct unique configurations, 
then $K(n) = \Omega(\log|\mathbb{P}|).$
% is a lower bound for the running time of the \maxsums\ in the decision tree model.
\end{theorem}

We shall observe that if every configuration were unique, then 
we could prove a lower bound of $\Omega(n\log n)$ since there exists  $n!$ configurations of size $n$.
However, there are configurations that are not unique like
the configuration $P=(1,2,1)$. In fact,
assume that $A=(a_1,a_2,a_3)$ is an input for which $P$ is its unique configuration.
Then, we would have both $a_1>a_3$ and $a_2+a_3 > a_1+a_2$, which is not possible.

The following 
construction  shows 
that the number of unique configurations is indeed very large.

\begin{theorem}
\label{prop:exponentialfeasibles}
There exist  $\Omega\parens*{2^{n}}$ unique configurations.
\end{theorem}
\begin{proof} 
Fix a subset $S \subseteq \left \{ 4, 5 ,\ldots,n \right \}$
and define 
\setstretch{0.6}
$$
a^S_i = 
\begin{cases}
0 & \mbox{if } i =1 \\ \\
2 & \mbox{if } i = 2\\ \\
4n & \mbox{if } i = 3\\ \\
3 & \mbox{if }  i \geq 4 \mbox{ and } i \notin S \\ \\
1 & \mbox{if }  i \geq 4 \mbox{ and } i \in S \\
\end{cases}
\qquad 
p^S_j =
\begin{cases}
3 & \mbox{if } j \leq n-2 \mbox{ and } (j+ 2) \notin S \\ \\
2 & \mbox{if }j \leq n-2 \mbox{ and } (j+ 2) \in S \\ \\
n - j + 1& \mbox{if }  j \geq n-1
\end{cases}
$$
\setstretch{1.5}

We first show that for the input $A^S=(a_1^S,\ldots,a_n^S)$ the unique output configuration is
$P^S=(p^S_1,\ldots,p^S_n).$ For this we need to verify 
that for each $j=1, \dots, n,$ in $A$ there is only one maximum consecutive subsum of length $j$ and it starts at $p^S_j$ as given above.

%\setstretch{0.8}
%\begin{align*}
%a^S_i = 
%\begin{cases}
%0 & \mbox{if } i =1 \\ \\
%2 & \mbox{if } i = 2\\ \\
%4n & \mbox{if } i = 3\\ \\
%3 & \mbox{if }  i \geq 4 \mbox{ and } i \notin S \\ \\
%1 & \mbox{if }  i \geq 4 \mbox{ and } i \in S \\
%\end{cases}
%\end{align*}
%\setstretch{1.5}
%
%We can easily show
%the unique output  for $A^S$ is the configuration $P^S=(p^S_1,\ldots,p^S_n)$ given by
%
%\setstretch{0.8}
%\begin{align*}
%p^S_j =
%\begin{cases}
%3 & \mbox{if } j \leq n-2 \mbox{ and } (j+ 2) \notin S \\ \\
%2 & \mbox{if }j \leq n-2 \mbox{ and } (j+ 2) \in S \\ \\
%n - j + 1& \mbox{if }  j \geq n-1
%\end{cases}
%\end{align*}
%\setstretch{1.5}

The statement is trivially true for  $j=n$ and for $j=n-1$ since in the latter case it is enough to observe that
$0= a_1 < a_2.$ 

Since $a_3 = 4 n > \sum_{i \neq 3} a_i$ it follows easily that for any $1 \leq k \leq n-2$ 
a maximum consecutive sum of length $k$  must contain the element $a_3$. Any such sum 
will then start at one of the first 3 elements. However, since any element other than $a_1$ is greater than 
zero, it also follows that a maximum consecutive sum of size $k \in \{1, \dots, n-2\}$ must start at $a_2$ or $a_3.$
In formula, for $1 \leq k \leq n-2,$ let $s_k$ denote the maximum consecutive sum of size $k$, then
$s_k \in \{\sum_{i=2}^{k+1} a_i , \sum_{i=3}^{k+2} a_i\}.$ Let $\Delta_k$ be the difference between the two possible 
values for $s_k$. We have $\Delta_k = \sum_{i=2}^{k+1} a_i - \sum_{i=3}^{k+2} a_i = a_2 - a_{k+2} = 2 - a_{k+2}.$
Thus, $\Delta_k > 0$ if $k \in S$ and $\Delta_k < 0$ otherwise, that is, $s_k =  \sum_{i=2}^{k+1} a_i$ if $k \in S$ and 
$s_k =  \sum_{i=3}^{k+2} a_i$ if $k \not \in S,$ which proves the above claim. The uniqueness of the configuration 
follows from the fact that all the inequalities in the above arguments are strict.  

Therefore, for each $S \subseteq \{4,\dots,n\}$
we obtain an unique output configuration.
Since there are 
$2^n / 8$ distinct subsets of  $\{4,\dots,n\}$ and each of them corresponds to a distinct unique configuration we have the desired result.

%We have shown that there is a one-to-one correspondence between the sequence of bits representing $S \subseteq \{4,\dots,n\}$
%and the sub-configuration $(p^S_2, \dots, p^S_{n-2}) \in \{2,3\}^{n-3}$ where as shown above $p^S_i = 2$ iff $i+2 \in S.$ Since %there are 
%$2^n / 8$ distinct subsets of  $\{4,\dots,n\}$ and each of them corresponds to a distinct unique configuration we have the %desired result. 
\end{proof}

\smallskip

The existence of at least exponentially many (in $n$) configurations supports our approach and 
is an indication of its potential. 
However, this result combined with Theorem \ref{theo:lowerbound} is still not enough 
to  obtain a non trivial (superlinear) lower bound for the \maxsums.
This motivated us to enrich our analysis by  empirically exploring the number of unique configurations.

\section{Empirical evidences that \maxsums\ requires $\Omega(n\log n)$ time}
\label{sec:Experimental}

In order to  count the number of unique configurations
it is important  to decide whether a given configuration $P$ is unique or not.
For that  we test whether there 
exists an input sequence $A=(a_1,\ldots,a_n)$ for which $P$ is its unique output configuration.
For $i \in \{1,\ldots,n-1\}$ and a configuration $P=(p_1, \dots, p_n)$, let ${\cal Q}(P,i)$ be the  following set of inequalities:
$$
\sum_{k = p_i}^{p_i + i - 1} a_k >  \sum_{ k = j } ^{j + i - 1} a_k \,\,\,\,\,\,\,\,\,\,\,\,  \mbox{ for }  j=1,\ldots,n-i+1 \mbox{ and } j \neq p_i 
$$
It is easy to realize that the contiguous subsequence of length $i$  that starts at position $p_i$ of $A$ has sum larger
than the sum of any  other contiguous subsequence of length $i$ if and only if
the  point  $A=(a_1,\ldots,a_n) \in R^n $ satisfies the above set of inequalities.
Thus,  $P$ is a unique configuration if and only if the set of inequalities 
 ${\cal Q}(P,1) \cup {\cal Q}(P,2) \cup \cdots \cup {\cal Q}(P,n-1)$ has a feasible solution.
In our experiments we employed a linear programming solver to perform this verification.

In order to speed up our computation we also employ a sufficient condition for the non-uniqueness of a configuration, which is 
provided by the following proposition.
% whose proof is deferred to the appendix.
%Before describing our evidences, we
%discuss an important property of the unique configurations, which
% allows us to significantly speed up our experiments and, as a consequence,
%to provide   empirical evidences for larger values of $n$.
%This property is given by Proposition~\ref{prop:adj1} (see the appendix for a proof).

\smallskip

\begin{proposition}[Non-Adjacency Property]
\label{prop:adj1}
Let $P=(p_1,\ldots,p_n)$ be an output configuration. If there exist $i,j \in \{1,\ldots,n\}$
such that $p_j = p_i + i$, then $P$ is not unique.
\end{proposition}
\begin{proof}
Let $i,j$ be such that $p_j = p_i + i$. In addition, assume that there is an input $A$ for which $P$ is its unique configuration. This means that
for any $r \neq i$ and $u \neq j$ we have
\begin{equation} \label{eq:basis}
\sum_{s = r}^{r+i-1} a_s < \sum_{s = p_i}^{p_i+i-1} a_s \qquad  \sum_{s = u}^{u+j-1} a_s < \sum_{s = p_j}^{p_j+j-1} a_s
\end{equation}

Let $m = \min \{i, j\}$  and let's split the sequence $a_{p_i} \, a_{p_i+1} \dots a_{p_i+i-1} \, a_{p_j} \, a_{p_j+1} \dots a_{p_j+j-1}$
into three parts, and let $R_1, R_2, R_3$ be the intervals of indices of the elements in these three parts defined as follows:
$$R_1 = \{p_i,  p_i+1,  \dots p_i+m-1\},  R_2 = \{p_i+m  , p_i+m+1,  \dots p_j+j-1-m\}, 
R_3 = \{p_j+j-m,  \dots p_j+j-1\}.$$

\noindent
{\em Case 1.}
If $i < j$  we have  $\sum_{s = p_i}^{p_i+j-1} a_s = \sum_{s \in R_1} a_s + \sum_{s \in R_2} a_s$ and 
 $\sum_{s = p_j}^{p_j+j-1} a_s = \sum_{s \in R_2} a_s + \sum_{s \in R_3} a_s.$ 
 Then applying the second inequality in (\ref{eq:basis}) we get
 $$\sum_{s \in R_1} a_s + \sum_{s \in R_2} a_s = \sum_{s=p_i}^{p_i+j-1} a_s < \sum_{s = p_j}^{p_j+j-1} a_s = \sum_{s \in R_2} a_s + \sum_{s \in R_3} a_s,$$
 which implies that 
 $\sum_{s = p_i}^{p_i+i-1} a_s = \sum_{s \in R_1} a_s < \sum_{s \in R_3} a_s = \sum_{s = p_j+j-i}^{p_j+j-1} a_s$ which is a contradiction to the 
 first inequality in (\ref{eq:basis}).

\noindent
{\em Case 2.}
If $i > j$  we have  $\sum_{s = p_i}^{p_i+i-1} a_s = \sum_{s \in R_1} a_s + \sum_{s \in R_2} a_s$ and 
 $\sum_{s = p_i+j}^{p_j+j-1} a_s = \sum_{s \in R_2} a_s + \sum_{s \in R_3} a_s.$ 
 Then applying the first inequality in (\ref{eq:basis}) we get
 $$\sum_{s \in R_1} a_s + \sum_{s \in R_2} a_s = \sum_{s= p_i}^{p_i+i-1} > \sum_{s = p_i+j}^{p_j+j-1} a_s = \sum_{s \in R_2} a_s + \sum_{s \in R_3} a_s,$$
 which implies that 
 $\sum_{s = p_i}^{p_i+j-1} a_s = \sum_{s \in R_1} a_s > \sum_{s \in R_3} a_s = \sum_{s \in p_j}^{p_j+j-1} a_s$ which is a contradiction to the 
 second inequality in (\ref{eq:basis}). 

\smallskip
Since in either case we have a contradiction, it follows that if $i,j$ are adjacent maxima than $P$ cannot be a unique configuration.
\end{proof}

\smallskip

For instance, consider the configuration $P=(5,1,3,4,1)$.
By taking $i=2$ and $j=3$, we have  $p_i+i =p_j$. Thus, we can conclude that  $P$ is not unique, as can be easily verified.

Unfortunately, this non-adjacency property does not completely characterize  the set of unique configurations  
because there exist configurations with no adjacent maximums that are not unique. 
For example, the configuration $P = (2, 4, 2, 1, 2, 1)$ has no adjacent maximums and is not unique.
In fact, if $A=(a_1,\ldots,a_6)$ is an input sequence for which $P$ is its unique configuration then we must
 have simultaneously: (i) $a_2  > a_4$ because of $p_1=2$; (ii) $a_4 + a_5  > a_1 + a_2$ because of  $p_2=4$; and (iii)
$a_1 + a_2  + a_3  + a_4 > a_2 + a_3  + a_4  + a_5$ because of $p_4=1$. 
However, this is impossible, since by summing the first two inequality and adding $a_3$ on both sides, 
we obtain a contradiction to the third inequality. Nonetheless, we can use the above condition to speed up our algorithms.

In order to exactly count
the number of unique configurations, we
explore the tree of   all permutations of  $\{1,\ldots,n\}$,
pruning the nodes that do not lead to an unique configuration.
In fact, assume that the
values of the positions $1,2,\ldots,i-1$ of
the  permutation $P$, under construction, are already  fixed. 
Then, for each $j$ that does not appear in the first $i-1$ positions, the procedure {\textsc{IsFeasible}$(P, i,j)$}, explained below, is called
to  verify whether it is possible to extend the partial permutation $P$ by
setting the value of position $i$ to $j$. 
If this is the case, the algorithm
set $p_i=j$ and it proceeds constructing the permutations.
Whenever we complete a permutation we increase the number of unique configurations.

%Algorithm \ref{alg:deterministic} shows the pseudo-code of our 
%method that recursively constructs all unique configurations of length $n$. 
%When \textsc{Deterministic-Counting}($P, i$) is called, the 
%values of the positions $1,2,\ldots,i-1$ of
%the  configuration $P$, under construction, are already  fixed. 
%Then, for each $j \in \{1,\ldots,n-i+1\}$, the procedure {\textsc{IsFeasible}$(P, i,j)$}, explained below, is called
%to  verify whether it is possible to extend the partial configuration $P$ by
%setting the value of position $i$ to $j$. 
%If this is the case, the algorithm
%set $p_i=j$ and it recursively calls \textsc{Deterministic-Counting}($P, i + 1$) to
% construct all unique configurations that are extensions of $P$.
%As soon as the algorithm set  the values of all  positions of $P$, it increments the global variable $count$ of unique %configurations.
%The procedure has to be initially called with $i=1$ and the variable $count$ shoud be initially set to 0.

The procedure {\textsc{Is-Feasible}$(P, i,j)$}
first  verifies if the subsequence of $A$ starting at position  $j$ is adjacent to some 
maximum subsequence that has already been fixed.
%This test can be done in $O(1)$  time by using a suitable data structure.
If this test is positive it rules out $j$ as a value for $p_i$ due to Proposition~\ref{prop:adj1}.
Otherwise, the procedure  verifies whether the set of inequalities  $Q(P,1)\cup \cdots \cup Q(P,i)$
is feasible and it returns TRUE or FALSE, accordingly.

\begin{algorithm}
{\small
\caption{\textsc{IsFeasible}($P, i,j$)}
\label{alg:isfeasible}
\begin{algorithmic}[1]
%\Statex
\IF{the subsequence of length $i$ starting at position $j$ is adajcent to the subsequence of length $k$ starting at $p_k$ for some $k<i$} 
  \STATE{\textbf{return} FALSE}
\ELSE
\STATE{$p_i =j$  \,\,\,\,\, $\%$ this is only to verify if this extension if feasible.}
\STATE{{\bf if} the set of inequalities  ${\cal Q}(P,1) \cup  \cdots \cup {\cal Q}(P,i)$is feasible}
\STATE{~~{\bf then} {\bf return} TRUE}
\STATE{{\bf else} {\bf return} FALSE}
%\If {the set of inequalities  ${\cal Q}(P,1) \cup  \cdots \cup {\cal Q}(P,i)$is feasible}
%  \State \textbf{return} TRUE
%\Else
%  \State \textbf{return} FALSE
%\EndIf
\ENDIF
\end{algorithmic}
}
\end{algorithm}
%\end{figure}

%Therefore, our brute force method is a simple back-tracking algorithm over the first half of the non-adjacent configurations. At level $i$ of the back-tracking algorithm we try an unexplored possibility for the maximum $P[i]$ and continue recursively, testing the feasibility of each possibility tried, pruning complete branches when a configuration is not unique any more and backtracking accordingly.

Table~\ref{tab:deterministicresults} presents the results obtained by the deterministic  approach.
We were able to determine the number of unique configurations up to $n=14$.
The results suggest a super exponential growth. Indeed, notice the growth of  the ratio between the number of unique configurations
and $(n/2)!$.
We selected the function $(n/2)!$ because it is a simple function 
whose  logarithm is $\theta (n \log n)$.

 % It can be noticed that the number of unique configurations grows super-exponentially. To see this, it's sufficient to compare the values against a factorial function like $\frac{n}{2}!$, which for $n=14$ is $5040$.

\begin{table}
\centering
\caption{The  number of unique configurations for $n=1,\ldots,14$ compared to the value
\label{tab:deterministicresults}
$n/2!$.}
% between the number of unique configurations and $(n/2)!$}
\begin{tabular}{ r c r c }
\hline
$n$ &  U(n) =  N\textsuperscript{\underline{o}} Unique Config. & $\frac{n}{2}!$ & Ratio $\frac{U(n)}{(n/2)!}$\\ 
\hline
$1$ & \numprint{1} & \numprint{0.8} & \numprint{1.25}x\\
$2$ & \numprint{2} & \numprint{1.0} & \numprint{2.00}x\\
$3$ & \numprint{4} & \numprint{1.3} & \numprint{3.07}x\\
$4$ & \numprint{12} & \numprint{2.0} & \numprint{6.00}x\\
$5$ & \numprint{36} & \numprint{3.3} & \numprint{10.90}x\\
$6$ & \numprint{148} & \numprint{6.0} & \numprint{24.66}x\\
$7$ & \numprint{586} & \numprint{11.6} & \numprint{50.51}x\\
$8$ & \numprint{2790} & \numprint{24.0} & \numprint{116.25}x\\
$9$ & \numprint{13338} & \numprint{52.3} & \numprint{255.02}x\\
$10$ & \numprint{71562} & \numprint{120.0} & \numprint{596.35}x\\
$11$ & \numprint{378024} & \numprint{287.9} & \numprint{1313.03}x\\
$12$ & \numprint{2222536} & \numprint{720.0} & \numprint{3086.85}x\\
$13$ & \numprint{12770406} & \numprint{1871.3} & \numprint{6824.34}x\\
$14$ & \numprint{78968306} & \numprint{5040.0} & \numprint{15668.31}x\\
\hline
\end{tabular}
\end{table}

All the executions required $27$ hours of CPU time under the following hardware and software specifications: Main Hardware Platform: Intel\textregistered~Core\texttrademark~i7~3960X, 3.30GHz CPU, 32GB RAM, 64-bit; OS: Windows 7 Professional x64; Compiler: Microsoft\textregistered Visual C\# 2010 Compiler version 4.0.30319.1; 
Solver: Gurobi Optimizer.%\cite{Gurobi:2013:GOI}.

%\begin{itemize}
%\item Main Hardware Platform: Intel\textregistered~Core\texttrademark~i7~3960X, 3.30GHz CPU, 32GB RAM, 64-bit.
%\item Operating System: Windows 7 Professional x64
%\item Compiler: Microsoft\textregistered Visual C\# 2010 Compiler version 4.0.30319.1
%\item Solver: Gurobi Optimizer\cite{Gurobi:2013:GOI}.
%\end{itemize}

In order to extend our analysis to larger instances
we then employed a probabilistic approach.

\subsection{A Probabilistic Approach}

The first idea for estimating
the number of unique configurations is to 
sample a large number $M$ of configurations and  test whether each of them is unique or not.
The number of unique configurations found over $M$ is an unbiased estimator for the number of unique configurations.
With this approach we managed to
obtain strong evidence of the super linear lower bound for $n$ up to 28.
To extend our range we followed a different approach.

%In the previous case, we explore the configuration space via a \emph{depth first} approach over the back-tracking tree of all possible configurations. For our new probabilistic approach we propose a more \emph{breadth first} like traversal. That is, for each level $i$ we first discover all values $P[i] = k$ for $k = 1, \ldots, n - i + 1$ such that $P$ is still unique. Let's say that exists $b_i$ feasible values, we call $b_i$ the \emph{branching factor} of our path at level $i$, then we choose randomly one of those values for  $P[i]$ and continue recursively with $P[i+1]$; if the method observes the  branching factors $b_1, b_2, \ldots, b_n$, in a root to leaf path, the method outputs \[X = \prod_{i = 1}^{n} b_i\] 

In the deterministic case, we explore the configuration space via a depth first search over the back-tracking tree of all possible configurations. In our probabilistic approach,  presented in 
Algorithm \ref{alg:bflv},
 we randomly traverse a  path in this  tree
that corresponds to a unique configuration.
Assume that we have already fixed the values for the positions $1,2,\ldots,i-1$ of
 the configuration $P$ that is under construction. Then, in order to set the value of $p_i$,
we construct a list $S$ of  all  values $j \in \{1,\ldots,n-i+1\}$ such that {\textsc{IsFeasible}($P, i,j$)} returns TRUE.
Let $b_i=|S|$ be the  \emph{branching factor} of our path at level $i$. Then, we  randomly  choose one of the  values in $S$ for  $p_i$ and continue  to set the values of $p_j$ for $j>i$; if the method observes the  branching factors $b_1, b_2, \ldots, b_n$, in a root to leaf path, then it outputs $\displaystyle{X = \prod_{i = 1}^{n} b_i}$
%\[X = \prod_{i = 1}^{n} b_i\] 
as a guess for the number of unique configurations.

\begin{algorithm}
{\small
\caption{\textsc{Branching-Product}($n$)}
\label{alg:bflv}
\begin{algorithmic}[1]
%\Statex
\STATE{$X \gets 1$ }
\FOR{$i\gets 1$ \textbf{to} $n$}
	\STATE{$S \gets \emptyset$}
	\FOR{$j\gets 1$ \textbf{to} $n - i + 1$}
	  \IF{\textsc{IsFeasible}($P, i,j$)} \label{li:randomprobing:isfeasible}  
	    \STATE{Add $j$ to the list $S$ of possible branchings} 
	  \ENDIF
	\ENDFOR  
	  \IF{$S$ is empty}
			\STATE{\textbf{return} $0$}
	  	\ELSE 
	      \STATE{$X \gets X \times |S|$}
	     \STATE{$p_i \gets $ value randomly selected from list $S$	\label{li:randomprobing:uniform}}  
		\ENDIF
\ENDFOR
\STATE{\textbf{return} $X$}
\end{algorithmic}
}
\end{algorithm}

The value $X$ can be used to estimate the number of unique configurations because $X$ is a sample of a random variable $\mathbf{X}$ whose expected value $\E[\mathbf{X}]$ is equal to the number of the unique configurations. In fact, let $\ell$ be a leaf located at depth $n$ of the backtracking tree, that is, $\ell$ corresponds to a unique configuration. The probability of reaching $\ell$ in our random walk  is $1/B(\ell)$, where $B(\ell)$ is the product of the branching factors in the path from the root of the tree to $\ell$. In addition, if $\ell$ is reached, the method outputs $B(\ell)$. Let $L$ be the set of leaves located at level $n$ in the backtracking tree. Thus we have that
\[\E[\mathbf{X}]=\sum_{\ell \in L} \frac{1}{B(\ell)}\times B(\ell) = |L|.\]

After coming up with this approach, we found out that it had already been proposed to study heuristics for backtracking algorithms~\cite{Kullman:2009:FBH,Knuth:1975:EFBP}.

We do not use directly the observed value $X$ to estimate the number of unique configurations; instead, we assume as a null hypothesis that $\E[\mathbf{X}]$, the number of unique configurations, is smaller than or equal to $(\frac{n}{2}!)$ and use the sampled value of $\mathbf{X}$ to reject this hypothesis with some level of confidence.

Let $c \ge 1$. Under the hypothesis that  $\E[\mathbf{X}] \leq (\frac{n}{2})!$, using 
Markov's inequality, it follows that:
$$ \Prob \left [\mathbf{X} \geq c\frac{n}{2}!\right ] \leq \Prob[X \geq c\E[\mathbf{X}]] \leq \frac{1}{c}$$
%\begin{align}
%\Prob \left [\mathbf{X} \geq c\E[\mathbf{X}] \right ] &\leq \frac{1}{c}\nonumber\\
%\Prob \left [\mathbf{X} \geq c\frac{n}{2}!\right ] &\leq \Prob[X \geq c\E[\mathbf{X}]] \nonumber\\
%\Prob \left [\mathbf{X} \geq c\frac{n}{2}! \right ] &\leq \frac{1}{c}\label{eq:markovgreater} \\
%\intertext{which implies that}
%\Prob \left [\mathbf{X} < c\frac{n}{2}! \right ] &\geq 1-\frac{1}{c}\label{eq:markovless}
%\end{align}
%
which implies that $\displaystyle{\Prob \left [\mathbf{X} < c\frac{n}{2}! \right ] \geq 1-\frac{1}{c}.}$
%$$\Prob \left [\mathbf{X} < c\frac{n}{2}! \right ] \geq 1-\frac{1}{c}.$$ %\label{eq:markovless}

%In other words, under the hypothesis that $\E[\mathbf{X}] \leq \frac{n}{2}!$, 
Therefore, if we sample $\mathbf{X}$ and find a value larger than $c\frac{n}{2}!$,
 %either we were "lucky" or the hypothesis is false and 
 we can reject the hypothesis and conclude that the number of unique configurations  is larger than $ \frac{n}{2}!$ with confidence of $1-\frac{1}{c}$. 
 %Then, if we reject the hypothesis, the only other possibility is that the number of unique configurations  is greater  than
 % $\frac{n}{2}!$

%For example, imagine that, for $n=10$, we sample a value $X=\numprint{46080}$; as $\numprint{46080}/\frac{10}{2}! =  \numprint{384}$, letting $c=384$ by (\ref{eq:markovless}), with a probability of at least $1-\frac{1}{384} \approx\numprint{99.73}\%$, we must have observed a value less than $384\frac{10}{2}!$, however we observed that value. Therefore, we can reject the hypothesis that $\E[X] \leq \frac{n}{2}!$ with $\numprint{99.73}\%$ level of confidence.  We know that in this case $\E[X] = \numprint{71562}$ which validates our rejection.

We can extend this approach by taking the maximum of $k$ samples. Let $X_1, \ldots, X_k$ be the values for $k$ samples of the random variable $\mathbf{X}$. 
Then, with the hypothesis  $\E[\mathbf{X}] \leq \frac{n}{2}!$ and 
using the above inequality we have
\begin{equation} 
\Prob \left [ \phantom{\frac{1}{2}} \max\{X_1,  \ldots, X_k\} \right. 
<  \left. c\frac{n}{2}! \right ] = 
\Prob \left [\bigwedge_{i=1}^k \left(X_i < c\frac{n}{2}! \right) \right ]  
=  \prod_{i=1}^k \Prob \left[X_i < c\frac{n}{2}! \right ] \geq  \left(1 - \frac{1}{c} \right)^k.
\end{equation}

%\begin{align}
%\Prob \left [\max\{X_1,  \ldots, X_k\} \geq c\frac{n}{2}! \right ] &=\Prob \left [X_1 \geq c\frac{n}{2}! \vee \ldots \vee X_k \geq c\frac{n}{2}! \right]  \nonumber\\
%&=1 - \Prob \left [ X_1 < c\frac{n}{2}! \wedge \ldots \wedge X_k < c\frac{n}{2}! \right ] \nonumber\\
%&=1 - \prod_{i=1}^{k}\Prob \left [X_i < c\frac{n}{2}! \right ] \nonumber \\
%&\leq 1 - \left(1 - \frac{1}{c}\right)^k\nonumber\\
%\Prob \left[ \max\{X_1,  \ldots, X_k\} \geq c\frac{n}{2}! \right] &\leq 1 - \left(1 - \frac{1}{c}\right)^k\\
%\intertext{Which implies that}
%\Prob \left [\max\{X_1,  \ldots, X_k\} < c\frac{n}{2}! \right ] &\geq \left(1 - \frac{1}{c}\right)^k
%\end{align}

%In other words, assuming that $\E[\mathbf{X}] \leq \frac{n}{2}!$, and using the $\max\{X_1, X_2, \ldots, X_k\}$, the probability %of all the samples being less or equal than $c\frac{n}{2}!$, for a large enough $c$, is very high, and if in those $k$ values we %find a value greater or equal than $c\frac{n}{2}!$, then we can reject the hypothesis and
%conclude that $\E[\mathbf{X}] \geq \frac{n}{2}!$ with confidence $\left( 1 - 1/c \right)^k.$

Thus,  if one of the values $X_1,X_2,\ldots,X_k$ is grater than or equal to $c\frac{n}{2}!$, then we can reject the hypothesis and
conclude that $\E[\mathbf{X}] \geq \frac{n}{2}!$ with confidence $\left( 1 - 1/c \right)^k.$

%For example, for $n=10$ and $k=10$, suppose that $\max\{X_1, \ldots, X_{10}\} = 378000$ which is equal to $\numprint{3150}\frac{10}{2}!$. In this case, 
%$\displaystyle{\Prob \left [\max\{X_1,  \ldots, X_{10}\} < c\frac{10}{2}! \right ] \! \geq \! \left(1 - %\frac{1}{\numprint{3150}}\right)^{10}
%\! \approx \numprint{99.68}\%.}$

%This implies that we can reject the hypothesis with a confidence of $\numprint{99.68}\%$ because we've found a value of $\numprint{3150}\frac{10}{2}!$. But, if the number of samples were $1000$ and the maximum remains the same, the confidence level would drop to $\numprint{72.79}\%$, which clearly shows the trade-off between the number values sampled and  the confidence level.

%It is important to note that even if we decide not to reject the hypothesis, for example if the confidence level were less than $50\%$, it does not validate the hypothesis that $\E[X] \leq \frac{n}{2}!$.

In our experiments we sampled $1000$ values of $\mathbf{X}$ for different values of $n$.
The choice of $1000$ for $k$ was to guarantee a reasonable CPU time.
Let $max(n,1000)$ be the maximum value found in the $\numprint{1000}$ samples and let 
$c_n = \lfloor max(n,1000) / \frac{n}{2}! \rfloor$.
Table~\ref{tab:factorUnbiasedEstimator} shows the probability of $\Prob[\max\{X_1, \ldots, X_{1000}\} < c_n\frac{n}{2}!]$ assuming that $\E[\mathbf{X}]\leq \frac{n}{2}!$ for configurations up to size $n=100$. 
 This probability also expresses our confidence to reject the hypothesis because in fact we've found a value equal to $c_n\frac{n}{2}!$. 
Based on these results we state the following conjecture.

\begin{conjecture}
The running time of \maxsums\ is $\Omega(n \log n)$ in the decision tree model.
\end{conjecture}

\begin{table}
\centering 
\caption{ $\Prob \left [\max\{X_1, \ldots, X_k\} < c_n\frac{n}{2}! \right ]$ for $k = \numprint{1000}$}
\label{tab:factorUnbiasedEstimator}
\npthousandthpartsep{}
\begin{tabular}{ l r  r }
\hline
$n$ & $ c_n $ & $\Prob[\max_{1}^{k}\{X_i\} < c_n \frac{n}{2}!]$\\ 
\hline
$10$ & \numprint{6048} & \numprint{99.98346560846560}$\%$\\
$11$ & \numprint{23760} & \numprint{99.99579124579120}$\%$\\
$12$ & \numprint{38880} & \numprint{99.99742798353910}$\%$\\
$13$ & \numprint{439296} &  \numprint{99.99977236305360}$\%$\\
$14$ & \numprint{558835} &  \numprint{99.99982105636870}$\%$\\
$20$ & \numprint{372252672} & \numprint{99.99999973136530}$\%$\\
$30$ & \numprint{102827922078}  & \numprint{99.99999999902750}$\%$\\
$40$ & \numprint{4680410711674}  & \numprint{99.99999999997860}$\%$\\
$50$ & \numprint{69590788615385}   & \numprint{99.99999999999860}$\%$\\
$60$ & \numprint{562841769233371}  & \numprint{99.99999999999980}$\%$\\
$70$ & \numprint{136904322455757}   & \numprint{99.99999999999930}$\%$\\
$80$ & \numprint{87399891508924}   & \numprint{99.99999999999890}$\%$\\
$90$ & \numprint{73279283017}   & \numprint{99.99999999863540}$\%$\\
$100$ & \numprint{204252401}  & \numprint{99.99999951040970}$\%$\\
\hline
\end{tabular}
\end{table}

%\section*{\large References} 
%We have described an approach  to prove lower bounds for the \maxsums\ and the $(min,+)$-convolution problem
%in the decision tree model. Using this approach we gave probabilistic evidence supporting the 
%conjecture that  both problems require $\Omega(n \log n)$ computational steps in the decision tree model of computation.
%
%Moreover, we believe that the techniques employed here can be of independent interest 
%in investigating theoretical lower bounds on other related computational problems.

%We show how a l

%We've given a new technique specifying a restricted input case (inputs with unique solutions or unique configurations) and we've shown that a lower bound over the number of unique configurations is sufficient to derive a lower bound for the \maxsums.
%We've presented empirical evidences to show that the number of unique configurations is larger than $\frac{n}{2}!$.  Thus, these %observations lead us to conjecture that the lower bound for the \maxsums\ is $\Omega(n\lg n)$.  

%\bibliography{paper}
%\bibliographystyle{ws-ijfcs}
%\bibliography{paper}

\begin{thebibliography}{1}

%\bibitem{Allison:2003:LBILNSI}  L. Allison, Longest biased interval and longest non-negative sum interval, Bioinformatics 19(10),1294-1295, 2003.
\bibitem{Golnaz:2012:BJSM}   G.~Badkobeh, G.~Fici, S.~Kroon, Z.~Lipt\'ak, Binary jumbled string matching for highly run-length compressible texts,  IPL 113:604-608,2013.

%%\bibitem{BaeTakaoka:2005:IAK}   S.E.Bae and T. Takaoka, Improved algorithms for the $k$-maximum subarray problem for small $k$, 
%%Proc. of COCOONÕ05, LNCS 3595 , pp. 621-631, 2005.
%%\bibitem{BaeTakaoka:2004:APKMS}  S.E. Bae and T.Takaoka, Algorithms for the problem of $k$-maximum sums and a vlsi algorithm for the $k$-maximum subarrays problem., ISPAN , pp. 247-253, 2004.
%\bibitem{BaeTakaoka:2006:IAK}  S.E. Bae and T.Takaoka, Improved algorithms for the $k$-maximum subarray problem, 
%Comput. J. 49(3), 358-374, 2006.
%\bibitem{BengtssonChen:2004:EAKMS}  F. Bengtsson and J. Chen, Efficient algorithms for $k$-maximum sums, Proc.\ of ISAAC 2004, LNCS 3341,  pp. 137-148, 2004.
\bibitem{BenOr83a}
Ben-Or, Lower bounds for algebraic computation trees, in: STOC: ACM Symposium
  on Theory of Computing (STOC), 1983.
%\bibitem{Bentley:1986:PP}   J. L. Bentley, Programming pearls (Addison-Wesley, 1986).
\bibitem{Bergkvist:2005:FAFDSED}   A. Bergkvist, P. Damaschke, Fast algorithms for finding disjoint subsequences
with extremal densities, Pattern Recognition 39, 2281-2292, 2006.
\bibitem{conf/esa/BremnerCDEHILT06}   D. Bremner, T.M. Chan, E.D. Demaine, J.\ Erickson, F.\ Hurtado, J.\ Iacono, S.\ Langerman, M.\ Patrascu, P.\ Taslakian, Necklaces, convolutions, and X+Y. Algorithmica 69, 294-314,2014.
%CoRR, abs/1212.4771, 2012. See also Proc. of ESA'06, LNCS 4168, pp. 160-171, 2006.
%\bibitem{Brodal:2007:LTA}   G. S. Brodal and A. G. Jorgensen, A linear time algorithm for the k maximal sums problem, Proc. of  MFCSÕ07, 
%pp. 442-453, 2007.
\bibitem{BurcsiIJFCS}   P. Burcsi, F. Cicalese, G. Fici and Zs. Lipt\'ak, Algorithms for jumbled pattern matching in strings, 
IJFCS, 23, 357-374, 2012.
\bibitem{Burcsi:2012:OAJPMS}  P. Burcsi, F. Cicalese, G. Fici and Zs. Lipt\'ak, On approximate jumbled pattern matching in strings, 
Th. of Comp. Systems 50(1), 35-51, 2012.
%\bibitem{BurcsiCFL:2010:OTA}  P. Burcsi, F. Cicalese, G.Fici and Zs. Lipt\'ak, On table arrangements, scrabble freaks, and jumbled pattern matching, Proc. of Fun with Algorithms, LNCS 6099, pp. 89-101, 2010.
\bibitem{ChanSTOC2015}
T.~Chan, M.~Lewenstein, Clustered Integer 3SUM via Additive Combinatorics, in Proc. of STOC 2015.
\bibitem{ChanICALP2014}
T.~Chan, A.~Amir, M.~Lewenstein, N.~Lewenstein, On hardness of jumbled
  indexing, in: ICALP, 2014.
%\bibitem{ChengChengTienChao:2005:IAK}   C.-H. Cheng, K.-Y. Chen, W.-C. Tien and K.-M. Chao, Improved algorithms for the $k$ maximum-sums problems, Proc.\ of  ISAACÕ05, LNCS 3827, pp. 799-808.
\bibitem{Cicalese:2009:CFL} F. Cicalese, G. Fici and Zs. Lipt\'ak, Searching for jumbled patterns in strings, 
Proc. of the Prague Stringology Conference, pp. 105-117, 2009.
\bibitem{Cicalese:2012:AMCSP}  F. Cicalese, E. Laber, O. Weimann and R. Yuster, Near linear time construction of an approximate index for all maximum consecutive sub-sums of a sequence, Proc. of CPM2012, LNCS 7354, pp. 149-158, 2012.
\bibitem{Cieliebak}  M. Cieliebak, T. Erlebach, Z. Lipt\'ak, J. Stoye and E. Welzl, Algorithmic complexity of protein identification: combinatorics of weighted strings, Discrete Applied Mathematics 137(1), 27-46, 2004.
\bibitem{Epp89}   D. Eppstein, Efficient algorithms for sequence analysis with concave and convex  gap costs, 
 PhD thesis, Comp.\ Sc.\ Dept., Columbia Univ., 1989.
%\bibitem{Fan:2003:OAM}   T.H. Fan, S. Lee, H.I. Lu, T.S. Tsou, T.C. Wang and A. Yao, An optimal algorithm for maximum-sum segment and its application in bioinformatics, Proc. of CIAAÕ03, pp. 251-257, 2003.
%\bibitem{conf/nips/FelzenszwalbHK03}   P.F. Felzenszwalb, D.P. Huttenlocher and J.M. Kleinberg, Fast algorithms for large- state-space HMMs with applications to web usage analysis, Proc. of NIPS 2003.
\bibitem{Fukuda:2001:DMO}   T. Fukuda, Y. Morimoto, S. Morishita and T. Tokuyama, Data mining with optimized two-dimensional association rules, ACM Trans. Database Syst. 26, 179-213, 2001.
\bibitem{Travis:ESA13} T. Gagie, D. Hermelin, G. M. Landau, O. Weimann, 
 Binary Jumbled Pattern Matching on Trees and Tree-Like Structures, Proc.\ of  ESA 2013, LNCS 8125, pp. 517-528, 2013
%\bibitem{Giaquinta:2013:NAB}   E. Giaquinta and S.Grabowski, New algorithms for binary jumbled pattern matching,
%Inf. Process. Lett. 113, 538-542, 2013.
%\bibitem{Grenander:1978:PA}  U. Grenander, Pattern Analysis (New York : Springer-Verlag, 1978).
%\bibitem{Gurobi:2013:GOI}  I. Gurobi Optimization, Gurobi optimizer reference manual (2013).
%\bibitem{Threesomes} A. Gronlund and S. Pettie, Threesomes, Degenerates, and Love Triangles, in Proc. of FOCS 2014, 621-630, 2014.
%\bibitem{Huang:1994:AAIRDNA}   X. Huang, An algorithm for identifying regions of a dna sequence that satisfy a
%content requirement, Bioinformatics 10(3), 219-225, 1994.
\bibitem{Karp:1972:RACP} R. Karp, Reducibility among combinatorial problems, Complexity of Computer Computations, eds. R. Miller and J. Thatcher (Plenum Press, 1972), pp. 85-103.
%\bibitem{books/daglib/0015106}  J. M. Kleinberg and E. Tardos, Algorithm design (Addison-Wesley, 2006).
\bibitem{Knuth:1975:EFBP}  D. E. Knuth, Estimating the efficiency of backtrack programs, Mathematics of computation 29(129), 122-136, 1975.
\bibitem{Kullman:2009:FBH}   O. Kullmann, Fundaments of branching heuristics, Handbook of Satisfiability, Frontiers in Artificial Intelligence
and Applications 185, pp. 205-244, 2009.
%\bibitem{LinLee:2007:RAS}  T.-C.Lin and D.T. Lee, Randomized algorithm for the sum selection problem, Theor.
Comput. Sci. 377 (May 2007) 151-156.
%\bibitem{Mitzenmacher:2005:PCR}   M. Mitzenmacher and E. Upfal, Probability and Computing: Randomized Algorithms and Probabilistic Analysis (Cambridge University Press, New York, NY, USA, 2005).
%\bibitem{Moosa:2010:IPBS}   T. M. Moosa and M. S. Rahman, Indexing permutations for binary strings, Inf.
%Process. Lett. 110, 795-798, 2010.
\bibitem{Moosa:2012:STL}  T.M. Moosa and  M.S.Rahman, Sub-quadratic time and linear space data structures
for permutation matching in binary strings, J. of Discrete Algorithms 10, pp. 5-9, 2012.
%\bibitem{Perumalla:1995:PAMSMS}   K. S. Perumalla and N. Deo, Parallel algorithms for maximum subsequence and
%maximum subarray, Parallel Processing Letters 5, pp. 367Ð373, 1995.
%\bibitem{RuzzoTompa:1999:LTAFAMSS}  W. L. Ruzzo and M. Tompa, A linear time algorithm for finding all maximal scoring subsequences., ISMB, (AAAI, 1999), pp. 234-241, 1999.
\bibitem{Williams} R.~Williams, Faster all-pairs shortest paths via circuit complexity, in: STOC:
  ACM Symposium on Theory of Computing (STOC), 2014.
%\bibitem{Takaoka:2002:EFMSP}   T. Takaoka, Efficient algorithms for the maximum subarray problem by distance
%matrix multiplication, Electr. Notes Theor. Comput. Sci. 61 (2002) 191Ð200.
%\bibitem{Tamaki:1998:AMS}   H. Tamaki and T. Tokuyama, Algorithms for the maximum subarray problem based on matrix multiplication, Proc. of SODA Õ98, pp. 446Ð452, 1998.
\end{thebibliography}
\setstretch{1.0}

%\pagebreak
%\appendix

\end{document}